\documentclass[11pt]{article}
\usepackage{amssymb,amsmath,url,hyperref,fullpage}
\usepackage{xcolor}
\usepackage{amsthm}

\title{The $\alpha$-divergences associated with a pair of strictly comparable quasi-arithmetic means}

\author{Frank Nielsen\footnote{E-mail: {\tt Frank.Nielsen@acm.org}}. \\ Sony Computer Science Laboratories Inc.\\ Tokyo, Japan}

\date{}

\def\leftsup#1#2{{}^{#1}{#2}}

\def\calX{\mathcal{X}}
\def\calF{\mathcal{F}}
\def\eqdef{:=}
 \def\st{\ :\ }

\def\calX{\mathcal{X}}
\def\calP{\mathcal{P}}

\def\bbR{\mathbb{R}}

\def\KL{\mathrm{KL}}

\def\dmu{\mathrm{d}\mu}

\def\bbR{\mathbb{R}}

\def\leftsup#1{{{}^{#1}}}

\newtheorem{lemma}{Lemma}
\newtheorem{definition}{Definition}

\newtheorem{theorem}{Theorem}
\def\pow{\mathrm{pow}}
\newtheorem{corollary}{Corollary}

\def\dnu{\mathrm{d}\nu}

\begin{document}
\maketitle

\begin{abstract}
We generalize the family of $\alpha$-divergences using a pair of strictly comparable weighted means.
In particular, we obtain the $1$-divergence in the limit case $\alpha\rightarrow 1$
  (a generalization of the Kullback-Leibler divergence) and the $0$-divergence in the limit case  $\alpha\rightarrow 0$
	 (a generalization of the reverse Kullback-Leibler divergence). 
We state the condition for a pair of quasi-arithmetic means to be strictly comparable, and report the formula for 
the quasi-arithmetic $\alpha$-divergences and its subfamily of bipower homogeneous $\alpha$-divergences which belong to the Csis\'ar's $f$-divergences.
	Finally, we show that these generalized quasi-arithmetic $1$-divergences and $0$-divergences can be decomposed as the sum of generalized cross-entropies minus entropies, and rewritten as conformal Bregman divergences using monotone embeddings.
\end{abstract}


%%%%
\section{Introduction}
%%%%

%%%
\subsection{Statistical divergences}
%%%%
Consider a measurable space $(\calX,\calF)$ (where $\calF$ denotes the $\sigma$-algebra and $\calX$ the sample space)
 equipped with a positive measure $\mu$ (e.g., usually the Lebesgue measure or the counting measure).
The notion of {\em statistical dissimilarity}~\cite{Basseville-2013} $D[P:Q]=D_\mu[p_\mu:q_\mu]$ 
between two arbitrary probability measures with Radon-Nikodym (RN) densities $p_\mu=\frac{\mathrm{d}P}{\dmu}$ and $q_\mu=\frac{\mathrm{d}Q}{\dmu}$ with respect to $\mu$ is at the core of many algorithms in signal processing, information theory, information fusion, and machine learning among others. 
When those statistical dissimilarities are smooth, they are called {\em divergences}~\cite{IG-2016} in the literature.
The most renown statistical divergence rooted in information theory~\cite{CoverThomasIT-2012} is the {\em Kullback-Leibler divergence} (KLD):
\begin{equation}
\KL_\mu[p_\mu:q_\mu] := \int_\calX   p_\mu(x)\log\frac{p_\mu(x)}{q_\mu(x)} \dmu(x).
\end{equation}
Since the KLD is independent of the reference measure $\mu$, i.e., $\KL_\mu[p_\mu:q_\mu]=\KL_\nu[p_\nu:q_\nu]$ for 
$p_\mu=\frac{\mathrm{d}P}{\dmu}$ and $q_\mu=\frac{\mathrm{d}Q}{\dmu}$, and $p_\nu=\frac{\mathrm{d}P}{\dnu}$ and $q_\nu=\frac{\mathrm{d}Q}{\dnu}$ the RN derivatives with respect to another positive measure $\nu$, we write concisely in the remainder:
\begin{equation}
\KL[p:q]=\int p\log\frac{p}{q}\dmu,
\end{equation}
instead of $\KL_\mu[p_\mu:q_\mu]$.
 
The KLD belongs to a parametric family of {\em $\alpha$-divergences}~\cite{alphadiv-2010} $I_\alpha[p:q]$ for $\alpha\in\bbR$:
\begin{equation}
I_\alpha[p:q] := \left\{
\begin{array}{ll}
\frac{1}{\alpha(1-\alpha)}  \left(1-\int p^\alpha q^{1-\alpha}\dmu\right), & \alpha\in\bbR\backslash\{0,1\} \label{eq:alphanormalized}\\
I_1[p:q]=\KL[p:q], & \alpha=1\\
I_0[p:q]=\KL[q:p],& \alpha=0
\end{array}
\right..
\end{equation}

The $\alpha$-divergences extended to positive densities (not necessarily normalized) play a central role in information geometry~\cite{IG-2016}:
\begin{equation}
I_\alpha[p:q] := \left\{
\begin{array}{ll}
\frac{1}{\alpha(1-\alpha)}\int \left(\alpha p+(1-\alpha)q-p^\alpha q^{1-\alpha}\right)\dmu, & \alpha\in\bbR\backslash\{0,1\} \label{eq:alpha}\\
I_1[p:q]=\KL_e[p:q], & \alpha=1\\
I_0[p:q]=\KL_e[q:p],& \alpha=0
\end{array}
\right.,
\end{equation}
where $\KL_e$ denotes the {\em extended Kullback-Leibler divergence}:
\begin{equation}
\KL_e[p:q] := \int \left(p\log\frac{p}{q}+q-p\right)\dmu.
\end{equation}

The $\alpha$-divergences are asymmetric for $\alpha\not=0$ (i.e., $I_\alpha[p:q]\not=I_\alpha[q:p]$ for $\alpha\not=0$) but exhibit the following {\em reference duality}~\cite{Zhang-2004}: 
\begin{equation}
I_\alpha[q:p]=I_{1-\alpha}[p:q]=(I_\alpha)^*[p:q],
\end{equation}
where we   denoted by $D^*[p:q]\eqdef D[q:p]$, the {\em reverse divergence} for an arbitrary divergence $D$ 
(e.g., $I_\alpha^*[p:q]=I_\alpha[q:p]=I_{1-\alpha}[p:q]$).

The $\alpha$-divergences belong to the family of Csiz\'ar's $f$-divergences~\cite{Csiszar-1967} which are defined for a convex function $f$ satisfying by $f(1)=0$ by:
\begin{equation}
I_f[p:q] \eqdef  \int p f\left(\frac{q}{p}\right) \dmu.
\end{equation}

We have
\begin{equation}
I_\alpha[p:q] = I_{f_\alpha}[p:q],
\end{equation}
with
\begin{equation}
f_\alpha(u)=\left\{
\begin{array}{ll}
\frac{1}{\alpha(1-\alpha)} (\alpha+(1-\alpha)u-u^{1-\alpha}), & \alpha\in \alpha\in\bbR\backslash\{0,1\}\\
u-1-\log u, & \alpha=1\\
1-u+u\log u, & \alpha=-1
\end{array}
\right.
\end{equation}

In information geometry, $\alpha$-divergences (and more generally $f$-divergences) are {\em invariant divergences}~\cite{IG-2016}, and it is customary to rewrite the $\alpha$-divergences using $\alpha_A=1-2\alpha$ (i.e., $\alpha=\frac{1-\alpha_A}{2}$). 
Thus the {\em extended $\alpha_A$-divergence} is defined by 

\begin{equation}
\hat{I}_{\alpha_A}[p:q]=\left\{
\begin{array}{ll}
\frac{4}{1-\alpha_A^2}\int \left(\frac{1-\alpha_A}{2} p+\frac{1+\alpha_A}{2}q-p^{\frac{1-\alpha_A}{2} } q^{\frac{1+\alpha_A}{2} }\right)\dmu, & \alpha\in\bbR\backslash\{-1,1\} \label{eq:alphaA}\\
\hat{I}_1[p:q]=\KL_e[p:q], & \alpha=1\\
\hat{I}_{-1}[p:q]=\KL_e[q:p],& \alpha=-1
\end{array}
\right.,
\end{equation}
and the reference duality is expressed by $\hat{I}_{\alpha_A}[q:p]=\hat{I}_{-\alpha_A}[p:q]$.

A statistical divergence $D[\cdot:\cdot]$ when evaluated on densities belonging to a given parametric family $\calP=\{p_\theta \st \theta\in\Theta\}$ of densities are equivalent to a corresponding {\em contrast function}~\cite{Eguchi-1992}:
\begin{equation}
D_\calP(\theta_1:\theta_2) := D[p_{\theta_1}:p_{\theta_2}].
\end{equation}
Although quite confusing, those contrast functions have also been called recently divergences in the literature~\cite{IG-2016}.
Thus to disambiguate whether the divergence is a statistical divergence or a parameter divergence (i.e., contrast function), we choose to use the brackets for encapsulating arguments in statistical divergences and the parenthesis to encapsulate parameter arguments in divergences which are contrast functions.
 
A smooth divergence $D(\theta_1:\theta_2)$ induces a dualistic structure in information geometry~\cite{IG-2016}.
For example, the KLD on the family $\Delta$ of probability mass functions on a finite alphabet $\calX$  with equivalent contrast function a Bregman divergence induces a {\em dually flat space}~\cite{IG-2016}.
More generally, the $\alpha_A$-divergences on the probability simplex $\Delta$ induce the {\em $\alpha_A$-geometry}. 
 
The $\alpha$-divergences are widely used in information sciences, 
see~\cite{integrationalpha-2007,alphadiv-NMF-2008,alphadiv-malapportionment-2018,alphadiv-harmonicBayesian-2019,alphadiv-dictlearning-2019,alphadiv-exptest-2019} just to cite a few applications.
The singly-parametric $\alpha$-divergences have also been generalized to  bi-parametric families of divergences
 like the $(\alpha,\beta)$-divergences~\cite{IG-2016} or the $\alpha\beta$-divergences~\cite{sarmiento2019centroid}.

In this work, based on the observation that the term $\alpha p+(1-\alpha)q-p^\alpha q^{1-\alpha}$ (in the extended $I_\alpha[p:q]$  divergence for $\alpha\in(0,1)$ of Eq.~\ref{eq:alpha}) is a difference between a 
weighted arithmetic mean  $A_{1-\alpha}(p,q):=\alpha p+(1-\alpha)q$ and a weighted geometric mean $G_{1-\alpha}(p,q):=p^\alpha q^{1-\alpha}$, we investigate a generalization of $\alpha$-divergences with respect to  a pair of abstract strictly comparable weighted means~\cite{ConvexFunction-2006}.

%%%%
\subsection{Divergences and decomposable divergences}\label{sec:dec}
%%%%
A statistical divergence $D$ shall satisfy the following two axioms: 
\begin{description}
\item[D1. non-negativity.] $D[p:q]\geq 0$ for all densities $p$ and $q$,
\item[D2. identity of indiscernibles.]   $D[p:q]=0$ if and only if $p=q$ $\mu$-almost everywhere.
\end{description}
These axioms are a subset of the metric axioms since we do not consider the symmetry axiom nor the triangular inequality axiom of metric distances.
See~\cite{probabilitymetric-2002} for some common examples of {\em probability metrics} (e.g., total variation or Wasserstein metrics). 

A divergence $D[p:q]$ is said {\em decomposable}~\cite{IG-2016} when it can be written as an integral of a scalar divergence $d(\cdot,\cdot)$:
\begin{equation}
D[p:q] = \int d(p(x):q(x))\dmu(x),
\end{equation}
or $D[p:q] = \int d(p:q)\dmu$ for short.

The $\alpha$-divergences are decomposable divergences since we have
\begin{equation}
I_\alpha[p:q] = \int i_\alpha(p(x):q(x)) \dmu
\end{equation}
with the following {\em scalar $\alpha$-divergence}:
\begin{equation}
i_\alpha(a:b) := \left\{
\begin{array}{ll}
\frac{1}{\alpha(1-\alpha)}  \left(\alpha a+(1-\alpha)b - a^\alpha b^{1-\alpha}\right), & \alpha\in\bbR\backslash\{0,1\} \label{eq:scalaralphanormalized}\\
i_1(a:b)=  a\log\frac{a}{b}+b-a  & \alpha=1\\
i_0(a:b)=i_1(b:a),& \alpha=0
\end{array}
\right.
\end{equation}

%\tableofcontents

%%%
\subsection{Contributions and paper outline}
%%%%
The outline of the paper and the contributions are summarized as follows:

We define the generic $\alpha$-divergences in \S\ref{sec:genalphadiv} for two families of strictly comparable means (Definition~\ref{def:genalphadiv}).
Then section~\ref{sec:qamalphadiv} reports a closed-form formula (Theorem~\ref{thm:qamalpha}) for the quasi-arithmetic $\alpha$-divergences induced by two strictly comparable quasi-arithmetic means with monotonically increasing generators $f$ and $g$ such that $f\circ g^{-1}$ is {\em strictly convex and differentiable}.
In \S\ref{sec:limitdiv}, we study  the divergences $I_0$ and $I_1$ obtained in the limit cases when $\alpha\rightarrow 0$ and $\alpha\rightarrow 1$, respectively.
We obtain generalized Kullback-Leibler divergences when $\alpha\rightarrow 1$ and generalized reverse Kullback-Leibler divergences when $\alpha\rightarrow 0$, which can be factorized as generalized cross-entropies minus entropies.
In \S\ref{sec:repBreg}, we show how to express these generalized $I_1$-divergences and $I_0$-divergences as  conformal Bregman representational divergences and briefly explain their induced conformally flat statistical manifolds.
Section~\ref{sec:rspower} explicits the subfamily of bipower homogeneous $\alpha$-divergences which belong to the family of Csisz\'ar $f$-divergences~\cite{Csiszar-1967}.
Finally, Section~\ref{sec:concl} summarizes the work and present several opportunities for future research directions.

%%%
\section{The $\alpha$-divergences based on a pair of means}\label{sec:genalphadiv}
%%%%

%%% 
\subsection{The abstract $(M,N)$ $\alpha$-divergences}\label{sec:mnalphadiv}
%%%
The point of departure for generalizing the $\alpha$-divergences is to rewrite Eq.~\ref{eq:alpha} for $\alpha\in\bbR\backslash\{0,1\}$ as 

\begin{equation}
I_\alpha[p:q] = \frac{1}{\alpha(1-\alpha)}\int \left(A_{1-\alpha}(p:q)-G_{1-\alpha}(p:q)\right)\dmu,
\end{equation} 
where $A_\lambda$ and $G_\lambda$ for $\lambda\in(0,1)$ stands for the weighted arithmetic and geometric means, respectively:
\begin{eqnarray*}
A_\lambda(x,y)&=& (1-\lambda)x+\lambda y,\\
G_\lambda(x,y)&=& x^{1-\lambda}y^\lambda.
\end{eqnarray*}
We choose the convention $A_0(x,y)=x$ and $A_1(x,y)=1$ so that $\{A_\lambda(x,y)\}_{\lambda\in[0,1]}$ smoothly interpolates between $x$ ($\lambda=0$) and $y$ ($\lambda=1$).

In general, a {\em mean} $M(x,y)$ aggregates two values $x$ and $y$ of an interval $I$
to produce an intermediate quantity which satisfies the {\em innerness property}~\cite{Bullen-2013}: 
\begin{equation}\label{eq:absm}
\min\{x,y\} \leq M(x,y) \leq \max\{x,y\},\quad \forall x,y\in I.
\end{equation}
A mean is said {\em strict} if the inequalities of Eq.~\ref{eq:absm} are strict whenever $x\not=y$.
A mean $M$ is said {\em reflexive} iff $M(x,x)=x$ for all $x\in I$.
In the remainder, we consider $I=(0,\infty)$.
By using the unique dyadic representation of any real $\lambda\in(0,1)$ (i.e., $\lambda=\sum_{i=1}^\infty \frac{d_i}{2^i}$ with $d_i\in\{0,1\}$, a binary digit) , one can build a {\em weighted mean} $M_\lambda$ from any given mean $M$, see~\cite{ConvexFunction-2006} for such a construction.

By analogy to the $\alpha$-divergences, let us define
 the (decomposable) {\em $(M,N)$ $\alpha$-divergences} for a pair of weighted means $M_{1-\alpha}$ and $N_{1-\alpha}$ for $\alpha\in (0,1)$  as
\begin{equation}\label{eq:genalphadiv}
I_\alpha^{M,N}[p:q] \eqdef \frac{1}{\alpha(1-\alpha)} \int\left (M_{1-\alpha}(p:q)-N_{1-\alpha}(p:q)\right)\dmu.
\end{equation} 

The ordinary $\alpha$-divergences for $\alpha\in(0,1)$ are recovered as the $(A,G)$ $\alpha$-divergences:
\begin{equation}
I_\alpha^{A,G}[p:q] = \frac{1}{\alpha(1-\alpha)}\int\left (A_{1-\alpha}(p:q)-G_{1-\alpha}(p:q)\right)\dmu
 =I_{1-\alpha}[p:q]=I_\alpha[q:p]=I^*_\alpha[p:q].
\end{equation}

In order to define   generalized $\alpha$-divergences satisfying the axioms D1 and D2 of proper divergences, we need to characterize the class of acceptable means.
We give a definition strengthening the notion of comparable means in~\cite{ConvexFunction-2006}:

\begin{definition}[Strictly comparable weighted means]
A pair $(M,N)$ of means are said {\em strictly comparable} whenever $M_\lambda(x,y)\leq N_\lambda(x,y)$ for all $x,y\in (0,\infty)$ with equality if and only if $x=y$, and for all $\lambda\in(0,1)$.
\end{definition}

For example, the inequality of the arithmetic and geometric means states that $A(x,y)\geq G(x,y)$: 
Means $A$ and $G$ are comparable, denoted by  $A\geq G$.
Furthermore, the arithmetic and geometric weighted means are distinct whenever $x\not=y$:
Indeed, consider the equation $(1-\alpha)x+\alpha y=x^{1-\alpha}y^\alpha$ for $x,y>0$ and $x\not =y$.
By taking the logarithm on both sides, we get
\begin{equation}
\log\left( (1-\alpha)x+\alpha y \right) = (1-\alpha)\log x + \alpha \log y.
\end{equation} 
Since the logarithm is a strictly convex function, the only solution is $x=y$.
Thus $(A,G)$ is a pair of strictly comparable weighted means.

For a weighted mean $M$, define $\bar{M}_\alpha(x,y):=M_{1-\alpha}(x,y)$.
We are ready to state the definition of generalized $\alpha$-divergences:

\begin{definition}[$(M,N)$ $\alpha$-divergences]\label{def:genalphadiv}
The $(M,N)$ $\alpha$-divergences $I_\alpha^{M,N}[p:q]$ between two positive densities $p$ and $q$ for $\alpha\in(0,1)$
is defined for a pair of strictly comparable weighted means $M_\alpha$ and $N_\alpha$ with $M_\alpha\geq N_\alpha$ by: 
\begin{eqnarray}
I_\alpha^{M,N}[p:q] &\eqdef& \frac{1}{\alpha(1-\alpha)}\int\left(M_{1-\alpha}(p:q)-N_{1-\alpha}(p:q)\right)\dmu,\\
 &=& \frac{1}{\alpha(1-\alpha)}\int\left(\bar{M}_{\alpha}(p:q)-\bar{N}_{\alpha}(p:q)\right)\dmu.
\end{eqnarray}
\end{definition}

Using $\alpha=\frac{1-\alpha_A}{2}$, we can rewrite this divergence as 
\begin{eqnarray}
\hat{I}_{\alpha_A}^{M,N}[p:q] &\eqdef& \frac{4}{1-\alpha_A^2} 
\int \left(M_{\frac{1+\alpha_A}{2}}(p:q)-N_{\frac{1+\alpha_A}{2}}(p:q)\right)\dmu,\\
 &=& \frac{4}{1-\alpha_A^2}\int\left(\bar{M}_{\frac{1-\alpha_A}{2}}(p:q)-\bar{N}_{\frac{1-\alpha_A}{2}}(p:q)\right)\dmu.
\end{eqnarray}

A weighted mean $M_\alpha$ is said {\em symmetric} iff $M_\alpha(x,y)=M_{1-\alpha}(y,x)$.
When both the weighted means $M$ and $N$ are symmetric, we have the following {\em reference duality}~\cite{Zhang-2004}:
\begin{equation}
I_\alpha^{M,N}[p:q]=I_{1-\alpha}^{M,N}[q:p].
\end{equation}
We consider symmetric means in the remainder.

In the limit cases of $\alpha\rightarrow 0$ or $\alpha\rightarrow 1$, we define the {\em $0$-divergence} $I_0^{M,N}[p:q]$
 and the {\em $1$-divergence} $I_1^{M,N}[p:q]$, respectively, by
\begin{eqnarray}
I_0^{M,N}[p:q] &=& \lim_{\alpha\rightarrow 0} I_\alpha^{M,N}[p:q], \\
I_1^{M,N}[p:q] &=&\lim_{\alpha\rightarrow 1} I_\alpha^{M,N}[p:q] =I_0^{M,N}[q:p],
\end{eqnarray}
provided that those limits exist.

%%%
\subsection{The quasi-arithmetic $\alpha$-divergences}\label{sec:qamalphadiv}
%%%%

A {\em quasi-arithmetic mean} (QAM) is defined  for a continuous and strictly monotonic function $f:I\subset\bbR\rightarrow J\subset\bbR$ as:
\begin{equation}
M^f(x,y) \eqdef f^{-1}\left( \frac{f(x)+f(y)}{2}  \right).
\end{equation}
Function $f$ is called the generator of the quasi-arithmetic mean.
These strict and reflexive quasi-arithmetic means are  also called  Kolmogorov means~\cite{Kolmogorov-1930}, Nagumo means~\cite{Nagumo-1930} or de Finetti means~\cite{deFinetti-1931}, or  quasi-linear means~\cite{Inequalities-1988} in the literature.
These means are called quasi-arithmetic means because they can be interpreted as arithmetic means on the arguments $f(x)$ and $f(y)$: 
\begin{equation}
f(M^f(x,y))=\frac{f(x)+f(y)}{2}=A(f(x),f(y)).
\end{equation}
QAMs are strict, reflexive  and symmetric means. 
 
Without loss of generality, we assume strictly increasing functions $f$ instead of monotonic functions since $M^{-f}=M^f$.
Indeed, $M^{-f}(x,y)=(-f)^{-1}(-f(M_f(x,y)))$ and $((-f)^{-1} \circ (-f)) (u)=u$, the identity function.
%Notice that the set of strictly monotonic functions for a group with group law the functional composition, identity element the identity function and group inverse the reciprocal function.
Notice that the composition $f_1\circ f_2$ of two strictly monotonic increasing functions $f_1$ and $f_2$ is a strictly monotonic increasing function.
Furthermore, we consider $I=J=(0,\infty)$ in the remainder since we apply these means on positive densities.
Two quasi-arithmetic means $M^f$ and $M^g$ coincide if and only if $f(u)=ag(u)+b$ for some $a>0$ and $b\in\bbR$ see~\cite{Inequalities-1988}.
The quasi-arithmetic means were considered in the axiomatization of the entropies by R\'enyi to define the $\alpha$-entropies (see Eq. 2.11 of~\cite{Renyi-1961}).

By choosing $f_A(u)=u$, $f_G(u)=\log u$ or $f_H(u)=\frac{1}{u}$,  we obtain the   Pythagorean's arithmetic $A$, geometric $G$, and harmonic $H$ means, respectively:
\begin{itemize}
\item the {\em arithmetic mean} (A): $A(x,y)=\frac{x+y}{2}=M^{f_A}(x,y)$, 
\item the {\em geometric mean} (G): $G(x,y)=\sqrt{xy}=M^{f_G}(x,y)$,  and 
\item the {\em harmonic mean} (H): $H(x,y)=\frac{2}{\frac{1}{x}+\frac{1}{y}}=\frac{2xy}{x+y}=M^{f_H}(x,y)$.
\end{itemize}

More generally, choosing $f_{P_r}(u)=u^r$, we obtain  the parametric family of {\em power means} (also called {\em H\"older means}~\cite{Holder-1889}):
\begin{equation}
P_r(x,y)=\left(\frac{x^r+y^r}{2}\right)^{\frac{1}{r}}=M^{f_{P_r}}(x,y),\quad r\in\bbR\backslash\{0\}.
\end{equation}

In order to get a {\em smooth family} of power means, we define the geometric mean  in the limit case of $r\rightarrow 0$:
\begin{equation}
P_0(x,y)=\lim_{r\rightarrow 0} P_r(x,y)=G(x,y)=\sqrt{xy}.
\end{equation}

It is known that the positively homogeneous quasi-arithmetic means, i.e. $M^f(\lambda a,\lambda b)=\lambda M^f(a,b)$ for $\lambda>0$, coincide exactly with the family of power means.
The weighted QAMs are given by
\begin{equation}\label{eq:wqam}
M^f_\alpha(p,q)=f^{-1}\left((1-\alpha)f(p)+\alpha f(q))\right)= f^{-1}\left(f(p)+\alpha (f(q)-f(p))\right)=M^f_{1-\alpha}(q,p).
\end{equation}

The logarithmic mean $L(x,y)$ for $x>0$  and $y>0$:
$$
L(x,y)=\frac{y-x}{\log y-\log x}
$$
is an example of a homogeneous mean (i.e., $L(\lambda x,\lambda y)=\lambda L(x,y)$ for any $\lambda>0$) that is {\em not} a QAM.
Besides the family of QAMs, there exist many other families of means~\cite{Bullen-2013}:
For example, let us mention the {\em Lagrangean means}~\cite{LagrangeanQAM-2007} which intersect with the QAMs only for the arithmetic mean, or a generalization of the QAMs called the the {\em Bajraktarevi\'c means}~\cite{BajraktarevicQAM-2020}.

Let us strengthen a recent theorem of~\cite{WQAM-2010} (Theorem~1, 2010):

\begin{theorem}[Strictly comparable weighted QAMs]\label{thm:compwm}
The pair $(M^f,M^g)$ of quasi-arithmetic means obtained for two strictly increasing generators is strictly comparable provided that $f\circ g^{-1}$ is strictly convex.
\end{theorem}

\begin{proof}
Since $f\circ g^{-1}$  is strictly convex, it is convex, and therefore it follows from 
Theorem~1 of~\cite{WQAM-2010} that $M_\alpha^f\geq M_\alpha^g$ for all $\alpha\in[0,1]$.
Thus the very nice property of QAMs is that $M^f\geq M^g$ implies that $M^f_\alpha\geq M^g_\alpha$ for {\em any} $\alpha\in [0,1]$.

Now, let us consider the equation $M^f_\alpha(p,q)=M^g_\alpha(p,q)$ for $p\not =q$:
\begin{equation}
f^{-1}\left((1-\alpha)f(p)+\alpha f(q)\right)  = g^{-1}\left((1-\alpha)g(p)+\alpha g(q)\right).
\end{equation}
Since $f\circ g^{-1}$ is assumed strictly convex,  and $g$ is strictly increasing, we have $g(p)\not=g(q)$ for $p\not=q$, and we reach the following contradiction: 
\begin{eqnarray}
(1-\alpha)f(p)+\alpha f(q) &=& (f\circ g^{-1})\left((1-\alpha)g(p)+\alpha g(q) \right ),\\
&<& (1-\alpha) (f\circ g^{-1})(g(p)) + \alpha (f\circ g^{-1})(g(q)),\\
&<& (1-\alpha)f(p)+\alpha f(q).
\end{eqnarray}
Thus $M^f_\alpha(p,q)\not=M^g_\alpha(p,q)$ for $p\not =q$, and $M^f_\alpha(p,q)=M^g_\alpha(p,q)$ for $p=q$.
\end{proof}

Note that the $(A,G)$ $\alpha$-divergences (i.e., the ordinary $\alpha$-divergences) is a proper divergence satisfying both the properties D1 and D2 because $f_A(u)=u$ and $f_G(u)=\log u$, and hence $(f_A\circ f_G^{-1})(u)=\exp(u)$ is strictly convex on $(0,\infty)$.

Let us denote by $I_\alpha^{f,g}[p:q]\eqdef I_\alpha^{M^f,M^g}[p:q]$ the {\em quasi-arithmetic $\alpha$-divergences}.
Since the QAMs are symmetric means, we have $I_\alpha^{f,g}[p:q]=I_{1-\alpha}^{f,g}[q:p]$.

%%%
\subsection{Limit cases of $1$-divergences and $0$-divergences}\label{sec:limitdiv}
%%%%
We seek a closed-form formula of the limit divergence $\lim_{\alpha\rightarrow 0} I_\alpha^{f,g}[p:q]$ when $\alpha\rightarrow 0$.

\begin{lemma}
A first-order Taylor approximation of the quasi-arithmetic mean~\cite{MNskewJensenBregman-2017} $M_\alpha^f$ for a $C_1$ strictly increasing generator $f$ when $\alpha\simeq 0$ yields
\begin{equation}
M_\alpha^f(p:q) = p+\frac{\alpha(f(q)-f(p))}{f'(p)} + o(\alpha(f(q)-f(p))).
\end{equation}
\end{lemma}
 
\begin{proof} 
By taking the first-order Taylor expansion of $f^{-1}(x)$ at $x_0$ (i.e., Taylor polynomial of order $1$), we get:
\begin{equation}
f^{-1}(x) = f^{-1}(x_0) + (x-x_0) (f^{-1})'(x_0) + o(x-x_0).
\end{equation}
Using the property of the derivative of an inverse function:
\begin{equation}
(f^{-1})'(x)=\frac{1}{(f'(f^{-1})(x))},
\end{equation} 
it follows that the first-order Taylor expansion of $f^{-1}(x)$ is:
\begin{equation}
f^{-1}(x) = f^{-1}(x_0)+ (x-x_0) \frac{1}{(f'(f^{-1})(x_0))} + o(x-x_0).
\end{equation}

Plugging $x_0=f(p)$ and $x=f(p)+\alpha(f(q)-f(p))$, we get a {\em first-order approximation} of the 
weighted quasi-arithmetic mean $M_\alpha^f$ when $\alpha\rightarrow 0$:

\begin{equation}
M_\alpha^f(p,q)  =  p +  \frac{\alpha(f(q)-f(p))}{f'(p)} + o(\alpha(f(q)-f(p))).
\end{equation}

\end{proof}

Let us introduce the following bivariate function:
\begin{equation}\label{eq:Ef}
E_f(p,q) \eqdef \frac{f(q)-f(p)}{f'(p)}.
\end{equation}

Thus we obtain closed-form formula for the $I_1$-divergence and $I_0$-divergence:

\begin{theorem}[Quasi-arithmetic $I_1$-divergence and $I_0$-divergence]
The quasi-arithmetic $I_1$-divergence induced by two strictly increasing and differentiable functions $f$ and $g$ such that $f\circ g^{-1}$ is strictly convex is
\begin{eqnarray}
I_1^{f,g}[p:q]=\lim_{\alpha\rightarrow 1} I_\alpha^{f,g}[p:q)] &=&
\int \left( E_f(p,q) -   E_g(p,q) \right)\dmu\geq 0,\label{eq:I1gen}\\
&=& \int \left( \frac{f(q)-f(p)}{f'(p)} -  \frac{g(q)-g(p)}{g'(p)}\right)\dmu.
\end{eqnarray}
We have $I_0^{f,g}[p:q]=I_1^{f,g}[q:p]$.
\end{theorem}

\begin{proof}
Let us prove that $I_1^{f,g}$  is a proper divergence satisfying axioms D1 and D2.
Note that a sufficient condition for $I_1^{f,g}[p:q]\geq 0$ is to check that
\begin{eqnarray}
E_f(p,q) &\geq& E_g(p,q),\\
\frac{f(q)-f(p)}{f'(p)} &\geq& \frac{g(q)-g(p)}{g'(p)}.
\end{eqnarray}

If $p=q$ $\mu$-a.e. then clearly $I_1^{f,g}[p:q]=0$.
Consider $p\not=q$ (i.e., at some observation $x$: $p(x)\not=q(x)$).

We shall use the following property of a strictly convex and differentiable function $h$ for $x<y$ (sometimes called the chordal slope lemma, see~\cite{ConvexFunction-2006}):
\begin{equation}\label{eq:ineqtangent}
h'(x)  \leq \frac{h(y)-h(x)}{y-x} \leq h'(y).
\end{equation}

We consider $h(x)=(f\circ g^{-1})(x)$ so that $h'(x)=\frac{f'(g^{-1}(x))}{g'(g^{-1}(x))}$.
There are two cases to consider:
\begin{itemize}
\item $p<q$ and therefore $g(p)<g(q)$. Let $y=g(q)$ and $x=g(p)$ in Eq.~\ref{eq:ineqtangent}.
We have $h'(x)=\frac{f'(p)}{g'(p)}$ and $h'(y)=\frac{f'(q)}{g'(q)}$, and the double inequality of Eq.~\ref{eq:ineqtangent} becomes
$$
\frac{f'(p)}{g'(p)} \leq \frac{f(q)-f(p)}{g(q)-g(p)} \leq \frac{f'(q)}{g'(q)}.
$$
Since $g(q)-g(p)>0$ and $g'(p)>0$ and $f'(p)>0$, we get
$$
\frac{g(q)-g(p)}{g'(p)} \leq \frac{f(q)-f(p)}{f'(p)}. 
$$

\item $q<p$ and therefore $g(p)>g(q)$.
Then the double inequality of Eq.~\ref{eq:ineqtangent} becomes
$$
\frac{f'(q)}{g'(q)} \leq \frac{f(q)-f(p)}{g(q)-g(p)} \leq \frac{f'(p)}{g'(p)}
$$
That is,
$$
\frac{f(q)-f(p)}{f'(p)} \geq \frac{g(q)-g(p)}{g'(p)},
$$
since $g(q)-g(p)<0$.
\end{itemize}
Thus in both cases, we checked that $E_f(p(x),q(x)) \geq E_g(p(x),q(x))$.
Therefore $I_1^{f,g}[p:q]\geq 0$ and since the QAMs are distinct $I_1^{f,g}[p:q]=0$ iff $p(x)=q(x)$ $\mu$-a.e.
\end{proof}

We can interpret the $I_1$ divergences as generalized KL divergences, and define generalized notions of cross-entropies and entropies.
Since the KL divergence can be written as the cross-entropy minus the entropy, we can also decompose the 
$I_1$ divergences as follows:

\begin{eqnarray}
I_1^{f,g}[p:q] &=& \int \left( \frac{f(q)}{f'(p)} - \frac{g(q)}{g'(p)} \right) \dmu - \int \left( \frac{f(p)}{f'(p)} - \frac{g(p)}{g'(p)} \right)\dmu,\\
&=& h^{f,g}_\times(p:q) - h^{f,g}(p),
\end{eqnarray}
where $h^{f,g}_\times(p:q)$ denotes the {\em $(f,g)$-cross-entropy} (for a constant $c\in\bbR$):
\begin{equation}
h^{f,g}_\times(p:q)=\int \left(\frac{f(q)}{f'(p)} - \frac{g(q)}{g'(p)}\right)\dmu + c,
\end{equation}
and $h^{f,g}(p)$ stands for the {\em $(f,g)$-entropy} (self cross-entropy):
\begin{equation}
h^{f,g}(p)=h^{f,g}_\times(p:p)=\int \left( \frac{f(p)}{f'(p)} - \frac{g(p)}{g'(p)} \right) \dmu + c.
\end{equation}

We define the  generalized {\em $(f,g)$-Kullback-Leibler divergence}:
\begin{equation}
\KL_{f,g}[p:q] := h^{f,g}_\times(p:q) - h^{f,g}(p).
\end{equation}

When $f=f_A$ and $g=f_G$, we resolve the constant to $c=0$, and recover the ordinary Shannon cross-entropy and entropy:
\begin{eqnarray}
h^{f_A,f_G}_\times(p:q) &=& \int (q - p\log q)\dmu = h_\times(p:q),\\
h^{f_A,f_G}(p:q) &=& h^{f_A,f_G}_\times(p:p)= \int (p-p\log p)\dmu = h(p),
\end{eqnarray}
and we have the $(f_A,f_G)$-Kullback-Leibler divergence that is the extended Kullback-Leibler divergence:
\begin{equation}
\KL_{f_A,f_G}[p:q] = \KL_e[p:q]=h_\times(p:q)-h(p)=\int (p\log\frac{p}{q}+q-p) \dmu.
\end{equation}

Thus we have the $(f,g)$-cross-entropy and $(f,g)$-entropy expressed as
\begin{eqnarray}
h^{f,g}_\times(p:q) &=& \int \left(\frac{f(q)}{f'(p)} - \frac{g(q)}{g'(p)}\right)\dmu,\\
h^{f,g}(p) &=&  \int \left( \frac{f(p)}{f'(p)} - \frac{g(p)}{g'(p)} \right) \dmu.
\end{eqnarray}

In general, we can define the $(f,g)$-Jeffreys' divergence as:
\begin{eqnarray}
J^{f,g}[p:q] &=& \KL^{f,g}[p:q] + \KL^{f,g}[q:p].
\end{eqnarray}

Thus we define the quasi-arithmetic mean $\alpha$-divergences as follows:

\begin{theorem}[Quasi-arithmetic $\alpha$-divergences]\label{thm:qamalpha} 
Let $f$ and $g$ be two strictly continuously increasing and differentiable functions on $(0,\infty)$ such that $f\circ g^{-1}$ is strictly convex.
Then the quasi-arithmetic $\alpha$-divergences induced by  $(f,g)$ for $\alpha\in [0,1]$ is
\begin{equation}
I_\alpha^{f,g}[p:q]=\left\{
\begin{array}{ll}
\frac{1}{\alpha(1-\alpha)}\int \left(M_{1-\alpha}^f(p:q)-M^g_{1-\alpha}(p:q)\right)\dmu, & \alpha\in\bbR\backslash\{0,1\}.\\
I_1^{f,g}[p:q]=\int \left( \frac{f(q)-f(p)}{f'(p)} -   \frac{g(q)-g(p)}{g'(p)} \right)\dmu & \alpha=1,\\
I_0^{f,g}[p:q]=I_1^{f,g}(q:p)=\int \left( \frac{f(p)-f(q)}{f'(q)} -   \frac{g(p)-g(q)}{g'(q)} \right)\dmu,& \alpha=0.
\end{array}
\right.
\end{equation}
\end{theorem}

When $f(u)=f_A(u)=u$ ($M^f=A$) and $g(u)=f_G(u)=\log u$ ($M^g=G$), we get:
\begin{equation}
I_1^{A,G}[p:q]=\int \left( q-p-p\log\frac{q}{p}\right) \dmu=\KL_e[p:q]=I_1[p:q],
\end{equation}
the Kullback-Leibler divergence (KLD) extended to positive densities, and $I_0=\KL_e^*$ the reverse extended KLD.

We can rewrite the $\alpha$-divergence $I_\alpha^{f,g}[p:q]$ for $\alpha\in (0,1)$ as
\begin{equation}
I_\alpha^{f,g}[p:q]=
\frac{1}{\alpha(1-\alpha)}  \left(S_{1-\alpha}^f(p:q)-S^g_{1-\alpha}(p:q)\right), 
\end{equation}
where
\begin{equation}
S_\lambda^h(p:q) := \int M^h_\lambda(p:q) \dmu.
\end{equation}

Zhang~\cite{Zhang-2004} (p. 188-189) considered the $(A,M^\rho)$ $\alpha_A$-divergences:
\begin{equation}\label{eq:ZhangAlpha}
D^{\rho}_{\alpha}[p:q]= \frac{4}{1-\alpha^{2}} \int\left(\frac{1-\alpha}{2} p+\frac{1+\alpha}{2} q-\rho^{-1}\left(\frac{1-\alpha}{2} \rho(p)+\frac{1+\alpha}{2} \rho(q)\right)\right)  \dmu.
\end{equation}
The formula he obtained for $D^{\rho}_{\pm 1}(p:q)$:
\begin{equation}
{D}^{\rho}_{1}[p:q]=\int\left(p-q-\left(\rho^{-1}\right)^{\prime}(\rho(q))(\rho(p)-\rho(q))\right) \dmu= {D}^{\rho}_{-1}[q:p],
\end{equation}
 is in accordance with our generic formula of Eq.~\ref{eq:I1gen} since $(\rho^{-1}(x))'=\frac{1}{\rho'(\rho^{-1}(x))}$.
Notice that $A_\alpha\geq P^r_\alpha$ for $r\leq 1$: The arithmetic weighted mean dominates the weighted power means $P^r$ when $r\leq 1$.

Furthermore, by imposing the homogeneity condition $I_\alpha^{A,M^\rho}[\lambda p:\lambda q]=\lambda I_\alpha^{A,M^\rho}[p:q]$  for $\lambda>0$,  Zhang~\cite{Zhang-2004} obtained the class of the $(\alpha_A,\beta_A)$-divergences for $(\alpha_A,\beta_A)\in [-1,1]^2$:
\begin{equation}
D_{\alpha_A, \beta_A}[p:q] = \frac{4}{1-\alpha_A^{2}} \frac{2}{1+\beta_A} \int\left(\frac{1-\alpha_A}{2} p+\frac{1+\alpha_A}{2} q-\left(\frac{1-\alpha_A}{2} p^{\frac{1-\beta_A}{2}}+\frac{1+\alpha_A}{2} q^{\frac{1-\beta_A}{2}}\right)^{\frac{2}{1-\beta_A}}\right) \dmu.
\end{equation}

%%%%
\subsection{Generalized KL divergences as conformal Bregman divergences on monotone embeddings}\label{sec:repBreg}
%%%%

We can rewrite the generalized KLDs $I^{f,g}_1$ as a {\em conformal Bregman representational divergence}~\cite{NNconformal-2015,Ohara-2018,Ohara-2019} as  follows:

\begin{theorem}
The generalized KLDs $I^{f,g}_1$ divergences are conformal Bregman representational divergences:
\begin{equation}
I_1^{f,g}[p:q] = \int \frac{1}{f'(p)} B_F(g(q):g(p)) \dmu,
\end{equation}
with $F=f\circ g^{-1}$ a strictly convex and differentiable Bregman convex generator.
\end{theorem}

\begin{proof}
For the Bregman strictly convex and differentiable generator $F=f\circ g^{-1}$, we expand the following conformal divergence:
\begin{eqnarray}
\frac{1}{f'(p)} B_F(g(q):g(p)) &=& \frac{1}{f'(p)} \left( F(g(q)) - F(g(p)) - (g(q)-g(p))F'(g(p))  \right),\\
&=& \frac{1}{f'(p)} \left( (f(q)-f(p)) - (g(q)-g(p))\frac{f'(p)}{g'(p)}  \right),
\end{eqnarray}
since $(g^{-1}\circ g)(x)=x$ and $F'(g(x))=\frac{f'(x)}{g'(x)}$.
It follows that
\begin{eqnarray}
\frac{1}{f'(p)} B_F(g(q):g(p)) &=& \frac{f(q)-f(p)}{f'(p)} - \frac{g(q)-g(p)}{g'(p)},\\
&=& E_f(p,q)-E_g(p,q) = I_1^{f,g}[p:q].
\end{eqnarray}
Hence, we easily check that $I_1^{f,g}[p:q]=\int \frac{1}{f'(p)} B_F(g(q):g(p)) \dmu\geq 0$ since $f'(p)>0$ and $B_F\geq 0$.
\end{proof}

In general, for a functional generator $f$ and a strictly monotonic representational function $r$, we can define the representational Bregman divergence~\cite{repBreg-2009} $B_{f\circ r^{-1}}(r(p):r(q))$ provided that $F=f\circ r^{-1}$ is a Bregman generator (i.e., strictly convex and differentiable).

In~\cite{MNskewJensenBregman-2017}, a generalization of the Bregman divergences was obtained using the {\em comparative convexity} induced by two abstract means $M$ and $N$ to define $(M,N)$-Bregman divergences as limit of scaled  $(M,N)$-Jensen divergences. 
The skew $(M,N)$-Jensen divergences are defined for $\alpha\in (0,1)$ by:
\begin{equation}
J_{F,\alpha}^{M,N}(p:q) = \frac{1}{\alpha(1-\alpha)} \left( N_\alpha(F(p),F(q)))-F(M_\alpha(p,q)) \right),
\end{equation}
where $M_\alpha$ and $N_\alpha$ are weighted means that should be {\em regular}~\cite{MNskewJensenBregman-2017} (i.e., homogeneous, symmetric, continuous and increasing in each variable).
Then we can define the $(M,N)$-Bregman divergence as
\begin{eqnarray*}
B_{F}^{M,N}[p:q] &=& \lim_{\alpha\rightarrow 1^-}    J_{F,\alpha}^{M,N}(p:q),\\
&=& \lim_{\alpha\rightarrow 1^-}  \frac{1}{\alpha(1-\alpha)} \left( N_\alpha(F(p),F(q)))-F(M_\alpha(p,q)) \right).
\end{eqnarray*}

The formula obtained in~\cite{MNskewJensenBregman-2017} for the quasi-arithmetic means $M^f$ and $M^g$ and a functional  generator $F$ that is $(M^f,M^g)$-convex is:
\begin{eqnarray}
B^{f,g}_F(p:q) 
&=&  
\frac{g(F(p))-g(F(q))}{g'(F(q))} - \frac{f(p)-f(q)}{f'(q)} F'(q),\\
&=& \frac{1}{f'(F(q))} B_{g\circ F\circ f^{-1}}(f(p):f(q)) \geq 0.
\end{eqnarray}

This is a conformal divergence~\cite{NNconformal-2015} that can be written using the $E_f$ terms as:
\begin{equation}
B^{f,g}_F(p:q) = E_g(F(q),F(p))-E_f(q,p)F'(q).
\end{equation}

A function $F$ is $(M^f,M^g)$-convex iff $g\circ F\circ f^{-1}$ is (ordinary) convex~\cite{MNskewJensenBregman-2017}.

The information geometry induced by a Bregman divergence (or equivalently by its convex generator) is a dually flat space~\cite{IG-2016,EIG-2018}.
The dualistic structure induced by a conformal Bregman representational divergence is related to conformal flattening~\cite{Ohara-2018,Ohara-2019}.

Following the work of Ohara~\cite{Ohara-2018,Ohara-2019}, the {\em geometric divergence} $\rho(p,r)$ (a contrast function in affine differential geometry)  induced by a pair $(L,M)$ of strictly monotone smooth functions 
between two distributions $p$ and $r$ of the $d$-dimensional probability simplex $\Delta_d$ is defined by (Eq.~28 in~\cite{Ohara-2018}):
\begin{equation}
\rho(p:r) =  \frac{1}{\Lambda(r)} \sum_{i=1}^{d+1} \frac{L(p_i)-L(r_i)}{L'(r_i)} =  \frac{1}{\Lambda(r)} \sum_{i=1}^{d+1} E_L(r_i,p_i),
\end{equation}
where $\Lambda(r)= \sum_{i=1}^{d+1} \frac{1}{L'(p_i)} p_i$. 
Affine immersions~\cite{Kurose-1994} can be interpreted as special embeddings.

Let $\rho$ be a divergence (contrast function) and $(\leftsup{\rho}g,\leftsup{\rho}\nabla,\leftsup{\rho}\nabla^*)$ be the induced statistical manifold structure
with
\begin{eqnarray}
\leftsup{\rho}g_{ij}(p) &\eqdef& -(\partial_i)_p(\partial_j)_p\ \rho(p,q)|_{q=p},\\
\Gamma_{ij,k}(p) &\eqdef& -(\partial_i)_p(\partial_j)_p(\partial_k)_q\ \rho(p,q)|_{q=p},\\
\Gamma_{ij,k}^*(p) &\eqdef& -(\partial_i)_p(\partial_j)_q(\partial_k)_q\ \rho(p,q)|_{q=p},
\end{eqnarray}
where $(\partial_i)_s$ denotes the tangent vector at $s$ of a vector field $\partial_i$.

Consider a conformal divergence $\rho_\kappa(p:q)=\kappa(q)\rho(p:q)$ for a positive function $\kappa(q)>0$, called the conformal factor.
Then the induced statistical manifold~\cite{Eguchi-1992,IG-2016} $(\leftsup{\rho_\kappa}g,\leftsup{\rho_\kappa}\nabla,\leftsup{\rho_\kappa}\nabla^*)$ is $1$-conformally equivalent to $(\leftsup{\rho}g,\leftsup{\rho}\nabla,\leftsup{\rho}\nabla^*)$ and we have
\begin{eqnarray}
\leftsup{\rho_\kappa}g &=& \kappa\ \leftsup{\rho}g,\\
\leftsup{\rho}g(\leftsup{\rho_\kappa}\nabla_X Y,Z) &=& \leftsup{\rho}g(\leftsup{\rho}\nabla_X Y,Z) -d(\log \kappa)(Z)\leftsup{\rho}g(X,Y).
\end{eqnarray}
The dual affine connections $\leftsup{\rho_\kappa}\nabla^*$ and  $\leftsup{\rho}\nabla^*$ are projectively equivalent~\cite{Kurose-1994} (and $\leftsup{\rho}\nabla^*$ is said $-1$-conformally flat).

Conformal flattening~\cite{Ohara-2018,Ohara-2019} consists in choosing the conformal factor $\kappa$ such that $(\leftsup{\rho_\kappa}g,\leftsup{\rho_\kappa}\nabla,\leftsup{\rho_\kappa}\nabla)$ becomes a dually flat space~\cite{IG-2016} equipped with a canonical Bregman divergence.

Therefore it follows that the statistical manifolds induced by the $1$-divergence $I_1^{f,g}$ is a representational $1$-conformally flat statistical manifold.

%%%%
\section{The subfamily of homogeneous $(r,s)$-power $\alpha$-divergences}\label{sec:rspower}
%%%%

In particular, we can define the {\em $(r,s)$-power $\alpha$-divergences} 
from two power means $P_r=M^{\pow_r}$ and $P_s=M^{\pow_s}$ with $r>s$  (and $P_r\geq P_s$) with the family of generators $\pow_l(u)=u^l$.
Indeed, we check that $f_{rs}(u)\eqdef \pow_r\circ\pow_s^{-1}(u)=u^{\frac{r}{s}}$ is strictly convex on $(0,\infty)$ since $f_{rs}''(u)=\frac{r}{s}\left(\frac{r}{s}-1\right)u^{\frac{r}{s}-2}>0$ for $r>s$. Thus $P_r$ and $P_s$ are two QAMs which are both comparable and distinct.
Table~\ref{tab:E} lists the expressions of $E_r(p,q)\eqdef E_{\pow_r}(p,q)$ obtained from the power mean generators $\pow_r(u)=u^r$.

\begin{table}
$$
\begin{array}{l|l}
\mbox{Power mean} & E_r(p,q)\\ \hline 
P_r (r\in\bbR\backslash\{0\}) & \frac{q^r-p^r}{rp^{r-1}}\\ \hline
Q (r=2) & \frac{q^2-p^2}{2p}\\
A (r=1) & q-p\\
G (r=0) & p\log\frac{q}{p}\\
H (r=-1) & -p^2\left(\frac{1}{q}-\frac{1}{p}\right)= p-\frac{p^2}{q} \\ \hline
\end{array}
$$
\caption{Expressions of the terms $E_r$ for the family of power means $P_r$, $r\in\bbR$.\label{tab:E}}
\end{table}

We conclude with the definition of the {\em $(r,s)$-power $\alpha$-divergences}:

\begin{corollary}[power $\alpha$-divergences]\label{def:powerrsalphadiv}
Given $r>s$, the $\alpha$-power divergences are defined for $r>s$ and $r,s\not=0$ by
\begin{equation}
I_\alpha^{r,s}[p:q]=\left\{
\begin{array}{ll}
\frac{1}{\alpha(1-\alpha)}\int \left( (\alpha p^r+(1-\alpha) q^r)^{\frac{1}{r}} - (\alpha p^s+(1-\alpha) q^s)^{\frac{1}{s}}  \right)\dmu, & \alpha\in\bbR\backslash\{0,1\}.\\
I_1^{r,s}[p:q]=\int \left( \frac{q^r-p^r}{rp^{r-1}} - \frac{q^s-p^s}{sp^{s-1}} \right)\dmu & \alpha=1,\\
I_0^{r,s}[p:q]=I_1^{r,s}(q:p) & \alpha=0.
\end{array}
\right.
\end{equation}
\end{corollary}

When $r=0$, we get the following power $\alpha$-divergences for $s<0$:
\begin{equation}
I_\alpha^{r,s}[p:q]=\left\{
\begin{array}{ll}
\frac{1}{\alpha(1-\alpha)}\int \left( p^\alpha q^{1-\alpha} - (\alpha p^s+(1-\alpha) q^s)^{\frac{1}{s}}  \right)\dmu, & \alpha\in\bbR\backslash\{0,1\}.\\
I_1^{r,s}[p:q]=\int \left(p\log\frac{q}{p} - \frac{q^s-p^s}{sp^{s-1}} \right)\dmu & \alpha=1,\\
I_0^{r,s}[p:q]=I_1^{r,s}[q:p] & \alpha=0.
\end{array}
\right.
\end{equation}

When $s=0$, we get the following power $\alpha$-divergences for $r>0$:
\begin{equation}
I_\alpha^{r,s}[p:q]=\left\{
\begin{array}{ll}
\frac{1}{\alpha(1-\alpha)}\int \left( (\alpha p^r+(1-\alpha) q^r)^{\frac{1}{r}} - p^\alpha q^{1-\alpha}  \right)\dmu, & \alpha\in\bbR\backslash\{0,1\}.\\
I_1^{r,s}[p:q]=\int \left( \frac{q^r-p^r}{rp^{r-1}} - p\log\frac{q}{p}  \right)\dmu & \alpha=1,\\
I_0^{r,s}[p:q]=I_1^{r,s}[q:p] & \alpha=0.
\end{array}
\right.
\end{equation}

In particular, we get the following family of {\em $(A,H)$ $\alpha$-divergences}:
\begin{equation}
I_\alpha^{A,H}[p:q]=I_\alpha^{1,-1}[p:q]=\left\{
\begin{array}{ll}
\frac{1}{\alpha(1-\alpha)}\int \left( \alpha p+(1-\alpha) q  -\frac{pq}{\alpha q+(1-\alpha) p}  \right)\dmu, & \alpha\in\bbR\backslash\{0,1\}.\\
I_1^{1,-1}[p:q]=\int \left( q-2p+\frac{p^2}{q}\right)\dmu & \alpha=1,\\
I_0^{1,-1}[p:q]=I_1^{1,-1}(q:p) & \alpha=0.
\end{array}
\right.,
\end{equation}
 and the family of {\em $(G,H)$ $\alpha$-divergences}: 
\begin{equation}
I_\alpha^{G,H}[p:q]=I_\alpha^{0,-1}(p:q)=\left\{
\begin{array}{ll}
\frac{1}{\alpha(1-\alpha)}\int \left( p^{\alpha}q^{1-\alpha}-\frac{pq}{\alpha q+(1-\alpha) p} \right)\dmu, & \alpha\in\bbR\backslash\{0,1\}.\\
I_1^{0,-1}[p:q]=\int \left(   p\log\frac{q}{p} -p + \frac{p^2}{q} \right)\dmu & \alpha=1,\\
I_0^{0,-1}[p:q]=I_1^{0,-1}[q:p] & \alpha=0.
\end{array}
\right..
\end{equation}

The $(r,s)$-power $\alpha$-divergences for $r,s\not=0$ yield homogeneous divergences: 
$I_\alpha^{r,s}[\lambda p:\lambda q]=\lambda I_\alpha^{r,s}[p:q]$ for any $\lambda>0$ because the power means are homogeneous:
 $P^r_\alpha(\lambda x,\lambda y)=\lambda P^r_\alpha(x,y)=\lambda x P^r_\alpha\left(1,\frac{y}{x}\right)$.
Thus the $I^{r,s}_\alpha$-divergences are Csisz\'ar $f$-divergences~\cite{Csiszar-1967}
\begin{equation}
I^{r,s}_\alpha[p:q] = \int p(x)f_{r,s}\left(\frac{q(x)}{p(x)}\right) \dmu
\end{equation}
 for the generator
\begin{equation}
f_{r,s}(u)=\frac{1}{\alpha(1-\alpha)}(P^r_\alpha(1,u)-P^s(1,u)).
\end{equation}

Thus the family of $(r,s)$-power $\alpha$-divergences are {\em homogeneous divergences}:
\begin{equation}
I^{r,s}_\alpha[\lambda p:\lambda q] = \lambda I^{r,s}_\alpha[p:q],\quad\forall \lambda>0
\end{equation}

%%%
\section{Conclusion, discussion and perspectives}\label{sec:concl}
%%%

For two comparable strict means $M\geq N$, one can define the {\em $(M,N)$-divergence} as
\begin{equation}
I^{M,N}[p:q] := 4\int \left(M(p,q)-N(p,q)\right) \dmu.
\end{equation}

When the property of strict comparable means extend to their induced weighted means $M_\alpha$ and $N_\alpha$ (i.e., $M_\alpha\geq N_\alpha$), one can further define the 
{\em $(M,N)$ $\alpha$-divergences} for $\alpha\in (0,1)$:
\begin{equation}
I^{M,N}_\alpha[p:q]:=\frac{1}{\alpha(1-\alpha)}\int \left(M_{1-\alpha}(p,q)-N_{1-\alpha}(p,q)\right) \dmu,
\end{equation}
so that $I^{M,N}[p:q] = I^{M,N}_{\frac{1}{2}}[p:q]$.
When the weighted means are symmetric,  the reference duality holds (i.e., $I^{M,N}_\alpha[q:p]=I^{M,N}_{1-\alpha}[p:q]$), 
and we can define the $(M,N)$-equivalent of the Kullback-Leibler divergence, i.e., the $(M,N)$ $1$-divergence, as the limit case  (when it exists): 
$I_1^{M,N}[p:q]=\lim_{\alpha\rightarrow 1 } I^{M,N}_\alpha[p:q]$.

We proved that the quasi-arithmetic weighted means $M_\alpha^f$ and  $M_\alpha^g$ are strictly comparable whenever $f\circ g^{-1}$ is {\em strictly} convex.
In the limit cases of $\alpha\rightarrow 0$ and  $\alpha\rightarrow 1$, we reported a closed-form formula for the equivalent of the Kullback-Leibler divergence and the reverse Kullback-Leibler divergence.
We  reported closed-form formula for the quasi-arithmetic $\alpha$-divergences $I^{f,g}_\alpha(p:q)$ for $\alpha\in [0,1]$ 
(Theorem~\ref{thm:qamalpha}),
 and for the subfamily of homogeneous $(r,s)$-power $\alpha$-divergences $I^{r,s}_\alpha(p:q)$ induced by power means (Corollary~\ref{def:powerrsalphadiv}). The ordinary $(A,G)$ $\alpha$-divergences, the $(A,H)$ $\alpha$-divergences and the  $(G,H)$ $\alpha$-divergences are examples of 
$(r,s)$-power $\alpha$-divergences for $(r,s)=(1,0)$, $(r,s)=(1,-1)$ and $(r,s)=(0,-1)$, respectively.

Generalized $\alpha$-divergences may prove useful in reporting closed-form formula between parametric densities:
For example, consider the ordinary $\alpha$-divergences between two scale Cauchy densities $p_1(x)=\frac{1}{\pi}\frac{s_1}{x^2+s_1^2}$ 
and $p_2(x)=\frac{1}{\pi}\frac{s_2}{x^2+s_2^2}$: There is no obvious closed-form for the ordinary
$\alpha$-divergences but we can report easily a closed-form for the $(A,H)$ $\alpha$-divergences following the calculus reported in~\cite{GenBhatPe-2014}:
\begin{eqnarray}
I_\alpha^{A,H}[p_1:p_2] &=&\frac{1}{\alpha(1-\alpha)} \left(1-\int H_{1-\alpha}(p_1(x),p_2(x))\dmu(x)\right),\\
&=& \frac{1}{\alpha(1-\alpha)} \left(1- \frac{s_1s_2}{(\alpha s_1+(1-\alpha) s_2)s_{1-\alpha}} \right),
\end{eqnarray}
with $s_{\alpha}=\sqrt{\frac{\alpha s_{1} s_{2}^{2}+(1-\alpha) s_{2} s_{1}^{2}}{\alpha s_{1}+(1-\alpha) s_{2}}}$.
 In general, instead of considering ordinary  $\alpha$-divergences in applications, one may consider the $(r,s)$-power $\alpha$-divergences, and tune the three parameters $(r,s,\alpha)$ according to the various tasks (say, by cross-validation in supervised machine learning tasks). We note that the quasi-arithmetic means have been recently considered by Eguchi et al.~\cite{Eguchi-2016} to define a novel non-parametric dualistic structure of information geometry via generalizations of the $e$-geodesics and the $m$-geodesics. 

The elucidation the {\em $(f,g)$ $\alpha$-geometry} of these generalized $\alpha$-divergences is left for future work.
For the limit cases of $\alpha\rightarrow 0$ or of $\alpha\rightarrow 1$, we proved that the limit KL-type divergences amount to a conformal Bregman divergence on a monotone embedding, and briefly showed the connection with conformal divergences and conformal flattening.
The geometry of conformal flattening~\cite{Ohara-2018,Ohara-2019} and the relationships with the $(\rho,\tau)$-monotone 
embeddings~\cite{naudts2018rho} shall be further studied.

 Applications of $(f,g)$ $\alpha$-divergences to center-based clustering~\cite{MixedAlphaClustering-2014} and to generalized $\alpha$-divergences in positive-definite matrix spaces~\cite{IG-2016} shall also be considered in future work.
The quasi-arithmetic weighted means are convex if and only if the generators are differentiable with positive first derivatives with corresponding functions $-E_f$ of Eq.~\ref{eq:Ef} convex (Theorem~4 of~\cite{convexQAM-2019}, i.e., convexity of the quasi-arithmetic weighted means does {\em not} depend on the weights).
For example, when both quasiarithmetic means are convex means, the quasi-arithmetic $\alpha$-divergence is the difference of two convex mean functions, and the $k$-means centroid computation amounts to solve a {\em Difference of Convex} (DC) program which can solved by the smooth DC Algorithm, DCA, called the Convex-ConCave Procedure~\cite{BR-2011}.
Similarly, when $\alpha\in\{0,1\}$, we get a DC program since $I_\alpha^{f,g}$ writes as a difference of convex functions.

\bibliographystyle{plain}
\bibliography{MNalphadivergenceBIB}

\end{document}